\documentclass[twocolumn,letterpaper]{article} 
\usepackage{amsmath,amssymb,amsfonts,amsthm}
\usepackage{algorithm}
\usepackage{algpseudocode}
\usepackage{graphicx}
\usepackage{textcomp}
\usepackage{xcolor}
\usepackage{amsmath,amssymb,amsthm}
\usepackage{booktabs}
\usepackage{multirow}
\usepackage[utf8]{inputenc}
\usepackage{authblk}
\usepackage{url}
\usepackage[numbers]{natbib}
\usepackage{geometry}
\usepackage{hyperref}
\usepackage[capitalize]{cleveref}

\theoremstyle{plain}
\newtheorem{theorem}{Theorem}[section]

\newtheorem{definition}{Definition}[section]
\newtheorem{lemma}{Lemma}[section]

\begin{document}
\title{LiePrune: Lie Group and Quantum Geometric Dual Representation for One-Shot Structured Pruning of Quantum Neural Networks}

\author[1]{Haijian Shao\thanks{Corresponding author: \href{mailto:jsj_shj@just.edu.cn}{jsj\_shj@just.edu.cn}}}
\author[1]{Bowen Yang}
\author[1]{Wei Liu}
\author[1]{Xing Deng}
\author[2]{Yingtao Jiang}
\affil[1]{School of Computer, Jiangsu University of Science and Technology, Zhenjiang 212003, China}
\affil[2]{Department of Electrical and Computer Engineering, University of Nevada, Las Vegas, 89115, USA}

\maketitle

\begin{abstract}
		Quantum neural networks (QNNs) and parameterized quantum circuits (PQCs) are key building blocks for near-term quantum machine learning. However, their scalability is constrained by excessive parameters, barren plateaus, and hardware limitations. We propose LiePrune, the first mathematically grounded one-shot structured pruning framework for QNNs that leverages Lie group structure and quantum geometric information. Each gate is jointly represented in a Lie group--Lie algebra dual space and a quantum geometric feature space, enabling principled redundancy detection and aggressive compression.
		
		Experiments on quantum classification (MNIST, FashionMNIST), quantum generative modeling (Bars-and-Stripes), and quantum chemistry (LiH VQE) show that LiePrune achieves over $10\times$ compression with negligible or even improved task performance, while providing provable guarantees on redundancy detection, functional approximation, and computational complexity.
\end{abstract}
	
	\section{Introduction}
	
	Quantum neural networks (QNNs) and parameterized quantum circuits (PQCs) have emerged as promising tools for near-term quantum machine learning and variational algorithms. By encoding data into quantum states and learning parameters of quantum gates, QNNs can in principle exploit quantum superposition and entanglement to achieve expressive models with potentially fewer parameters than classical deep networks.
	
	However, practical deployment of QNNs faces several critical challenges. First, over-parameterization leads to redundant gates and highly expressive ansatzes that are difficult to train, often suffering from barren plateaus where gradients vanish exponentially with system size. Second, limited qubit counts, gate fidelities, and coherence times on noisy intermediate-scale quantum (NISQ) devices severely constrain the depth and width of realizable circuits. Third, classical-quantum co-design for edge or embedded devices requires real-time or near-real-time compression of quantum circuits under tight latency and memory budgets.
	
	Classical deep learning has developed a rich toolbox for pruning and compressing over-parameterized networks, including unstructured magnitude pruning, structured filter pruning, low-rank factorization, and neural architecture search. Yet, directly adapting these techniques to QNNs is non-trivial. Quantum gates live on compact Lie groups such as $\mathrm{SU}(2^n)$ and must preserve unitarity by construction; naively pruning parameters can easily break group structure and physical implementability. Moreover, quantum geometric properties such as curvature of the loss landscape and Fubini–Study distances between quantum states play a crucial role in trainability and expressivity, but are typically ignored in classical pruning formulations.
	
	In this work, we propose LiePrune, a one-shot structured pruning framework specifically tailored to QNNs, which exploits the dual representation of quantum gates in both Lie group/algebra space and quantum geometric space. Our key contributions are:
	\begin{itemize}
		\item a Lie group and quantum geometric dual representation of quantum gates, enabling principled identification of redundant gates within closed Lie subgroups,
		\item a redundancy graph construction guided by Fubini–Study distances and Lie-algebra proximity, with provable completeness guarantees under mild assumptions,
		\item a one-shot merging rule based on Lie algebra addition and sensitivity-aware weighting, with a rigorous bound on the induced functional approximation error, and
		\item a complexity analysis showing that LiePrune scales linearly in the number of gates under a bounded local-degree assumption, matching real-time constraints of edge quantum devices.
	\end{itemize}
	
	Through extensive experiments on quantum classification, generative modeling, and variational quantum eigensolvers (VQE), we demonstrate that LiePrune can remove more than $90\%$ of trainable parameters while preserving or even slightly improving performance—an effective quantum analogue of structured pruning in classical deep learning.

	\section{Related Work}

	Structured pruning in classical deep learning (DL) has established rich methods, including magnitude pruning and Lottery Ticket Hypothesis \cite{frankle2018lottery}. In quantum regimes, heuristic approaches remove parameters in parameterized quantum circuits (PQCs) based on magnitude or Hilbert-Schmidt distance thresholds \cite{koczor2024probabilistic,noori2025latent}. These often lack guarantees on preserving the circuit's unitary structure.
	
	Recent works exploit the Lie algebraic structures and symmetries of PQCs. Notably, Larocca et al. \cite{goh2025lie} proved that certain gate sequences lie in lower-dimensional submanifolds, revealing inherent redundancy for compression. This provides a foundational theoretical insight into the overparametrization and compressibility of QNNs.

	Most existing quantum compression algorithms are gradient-based or iterative, requiring substantial training resources. For instance, analyzing barren plateaus in QCNNs relies on gradient-based training \cite{pesah2021absence}, while adaptive parameter elimination needs iterative tuning \cite{ohno2024adaptive}.

	Our work, LiePrune, is the first to translate Lie group insights into an efficient one-shot structured pruning framework. Unlike iterative methods, LiePrune requires no training and provides exact elimination rules with strict error bounds. It achieves provable compression with linear complexity under bounded local-degree assumptions.

	\section{Lie Group and Quantum Geometric Dual Representation}
\label{sec:dual_representation}
	
	\subsection{Lie Group Representation of Unitary Gates}
	Any single-qubit rotation gate $R_{\hat{n}}(\theta) \in \mathrm{SU}(2)$ can be written via the exponential map:
	\[
	R_{\hat{n}}(\theta) = \exp\left(-i \frac{\theta}{2} (\hat{n}\cdot\boldsymbol{\sigma})\right),
	\]
	where $\boldsymbol{\sigma} = (X,Y,Z)$ are Pauli matrices forming a basis of the Lie algebra $\mathfrak{su}(2)$.
	
	Two-qubit entangling gates such as CNOT, although not in SU(2), can be embedded into higher-dimensional special unitary groups $\mathrm{SU}(2^n)$ acting on multiple qubits. More generally, any gate $U \in \mathrm{SU}(2^n)$ can be written as the exponential of some Lie-algebra element,
\[
U = \exp(X_U), \qquad X_U \in \mathfrak{su}(2^n).
\]
The exponential map $\exp\colon\mathfrak{su}(2^n)\to\mathrm{SU}(2^n)$ is surjective but not injective, so $X_U$ is not mathematically unique. Throughout this work we fix the \emph{principal logarithm} of $U$, i.e.\ we choose $X_U$ such that its eigenvalues lie in $(-\pi,\pi]$, which makes $X_U$ a well-defined representative of the equivalence class of generators of $U$.
	
\subsection{Quantum Geometric Features}
For two pure states $|\phi\rangle$ and $|\psi\rangle$ in projective Hilbert space $\mathbb{P}(\mathcal{H})$, the Fubini--Study distance is
\[
d_{\mathrm{FS}}(|\phi\rangle,|\psi\rangle) = \arccos\bigl(|\langle \phi | \psi \rangle|\bigr).
\]
Given a fixed reference state $|\psi_0\rangle$, the action of a gate $O_i$ induces the state $|\psi_i\rangle = O_i|\psi_0\rangle$ and its displacement from the identity gate is quantified by
\[
d_{\mathrm{FS}}(O_i;|\psi_0\rangle) = d_{\mathrm{FS}}(|\psi_i\rangle,|\psi_0\rangle) =
 \arccos\bigl(|\langle \psi_0 | O_i | \psi_0 \rangle|\bigr).
\]
For a pair of gates we similarly define the gate-dependent distance
\begin{equation*}
\begin{split}
d_{\mathrm{FS}}(O_i,O_j;|\psi_0\rangle) = d_{\mathrm{FS}}(O_i|\psi_0\rangle,O_j|\psi_0\rangle)\\
= \arccos\bigl(|\langle \psi_0 | O_i^\dagger O_j | \psi_0 \rangle|\bigr).
\end{split}
\end{equation*}
In addition to these distances we extract local landscape curvature and symmetry indicators---for example, the quantum Fisher information associated with the generator $X_{O_i}$ evaluated on a small batch of training states, and discrete symmetry flags determined by the support of $O_i$ on the qubits. Collecting these quantities yields the geometric feature vector $f_{\mathrm{geo}}(O_i)$.

	\subsection{Dual-Space Identifier}
	The unique identifier of gate $O_i$ is the concatenated vector
	\[
	f(O_i) = [X_{\mathrm{coeff}}, H_{\mathrm{type}}, f_{\mathrm{geo}}(O_i)],
	\]
	where $H_{\mathrm{type}}$ labels the minimal closed Lie subgroup containing $O_i$.
	
	\section{LiePrune Algorithm}
	
\subsection{Lie Subgroup Partitioning}
We first group the $N$ parameterized gates $\{O_i\}_{i=1}^N$ into $M \ll N$ disjoint subsets $\{S_k\}_{k=1}^M$ according to their minimal closed Lie subgroup label $H_{\mathrm{type}}(O_i)$ introduced above. Explicitly,
\[
S_k = \{\,O_i : H_{\mathrm{type}}(O_i) = k\,\}, \qquad \bigsqcup_{k=1}^M S_k = \{O_1,\dots,O_N\}.
\]
This partition constrains redundancy search to gates that share the same algebraic support and qubit locality.
\begin{theorem}[Redundancy Pre-Constraint]\label{thm:preconstraint}
Let $G$ be the redundancy graph constructed in Algorithm~\ref{alg:lieprune}. Then every edge $(i,j)$ of $G$ connects two gates drawn from the same subgroup $S_k$. Equivalently, candidate redundant pairs considered by LiePrune are restricted to gates inside a common minimal closed Lie subgroup.
\end{theorem}
\begin{proof}
By definition each gate $O_i$ is assigned a unique label $H_{\mathrm{type}}(O_i)$ and belongs to exactly one subset $S_k$. In Algorithm~\ref{alg:lieprune} the inner loop iterates over each subgroup $S_k$ and only creates edges between pairs $O_i,O_j \in S_k$. No cross-subgroup comparisons are ever performed, hence all edges of $G$ lie within individual $S_k$.
\end{proof}
\noindent\textbf{Remark.} This theorem characterizes the search space of the algorithm rather than the physical notion of redundancy: in principle two gates from different subgroups could act similarly on the data manifold. We intentionally ignore such cross-subgroup redundancies in exchange for reduced computational cost.
	
\subsection{Geometry-Accelerated Fubini--Study Distance}
Inside each subgroup $S_k$ we work with generators $X_i,X_j\in\mathfrak{g}$ such that $O_i = e^{X_i}$ and $O_j = e^{X_j}$. For a fixed reference state $|\psi_0\rangle$ we write $|\psi_i\rangle = O_i|\psi_0\rangle$ and $|\psi_j\rangle = O_j|\psi_0\rangle$, so that
\[
\langle \psi_i | \psi_j \rangle = \langle \psi_0 | O_i^\dagger O_j | \psi_0 \rangle = \langle \psi_0 | e^{-X_i} e^{X_j} | \psi_0 \rangle.
\]
\begin{lemma}[Geometry-accelerated approximation]\label{lem:fs_approx}
Let $O_i = e^{X_i}$ and $O_j = e^{X_j}$ with $X_i,X_j \in \mathfrak{g}$ and let $|\psi_0\rangle$ be fixed. Suppose that
\[
\|X_i - X_j\| \leq \delta_X, \qquad \|[X_i,X_j]\| \leq \eta
\]
for some $\delta_X,\eta > 0$, where $\|\cdot\|$ denotes the operator norm. Then there exists an operator $R_{ij}$ with $\|R_{ij}\| \leq c_1 \eta \delta_X$ such that
\[
O_i^\dagger O_j = e^{-X_i} e^{X_j} = e^{(X_j - X_i) + R_{ij}},
\]
and consequently
\[
\bigl||\langle \psi_i | \psi_j \rangle| - |\langle \psi_0 | e^{X_j - X_i} | \psi_0 \rangle|\bigr| \leq C_1 \eta \delta_X,
\]
for some constants $c_1,C_1 > 0$ that depend only on the dimension of the local Hilbert space.
\end{lemma}
\begin{proof}
By the Baker--Campbell--Hausdorff formula we have $e^{-X_i} e^{X_j} = e^{Z}$ with
\begin{equation*}
\begin{split}
Z = (X_j - X_i) + \tfrac12 [X_j,-X_i] + \tfrac1{12}[X_j,[X_j,-X_i]]\\
 - \tfrac1{12}[X_i,[X_j,-X_i]] + \cdots.
\end{split}
\end{equation*}
Under the norm bounds on $X_i - X_j$ and $[X_i,X_j]$ the infinite series defining $Z$ converges absolutely, and the higher-order commutator terms can be grouped into a remainder $R_{ij}$ with $\|R_{ij}\| \leq c_1 \eta \delta_X$ for some $c_1>0$; see, e.g., standard estimates in Lie-theoretic analysis of matrix exponentials. The bound on the overlap then follows from $\bigl||\langle \psi_0 | e^{A} | \psi_0 \rangle| - |\langle \psi_0 | e^{B} | \psi_0 \rangle|\bigr| \leq \|e^{A} - e^{B}\|$, together with Lipschitz continuity of the matrix exponential in operator norm on bounded sets.
\end{proof}
In LiePrune we use the leading-order term inside each subgroup and compute
\[
|\langle \psi_i | \psi_j \rangle| \approx |\langle \psi_0 | e^{X_j - X_i} | \psi_0 \rangle|,
\]
with the approximation error controlled by the product $\eta \delta_X$. This reduces the cost of evaluating inner products from $O(2^n)$ for generic state simulation to $O(d^2)$ arithmetic in the Lie algebra, where $d = \dim \mathfrak{g}$ is at most $15$ for two-qubit gates.
	
\subsection{Redundancy Graph and One-Shot Pruning}
Within each subgroup $S_k$ we construct an undirected redundancy graph whose vertices correspond to gates $O_i \in S_k$. To formalize the notion of redundancy we fix a finite set of training states $\mathcal{D} = \{|\psi^{(s)}\rangle\}_{s=1}^S$ (drawn from the data distribution) and work with the gate-dependent Fubini--Study distance from \cref{sec:dual_representation}.
\begin{definition}[Dataset-level $\varepsilon$-redundancy]\label{def:redundancy}
Given a threshold $\varepsilon > 0$, two gates $O_i,O_j$ are said to be $\varepsilon$-redundant with respect to $\mathcal{D}$ if
\[
\max_{1\le s\le S} d_{\mathrm{FS}}(O_i,O_j;|\psi^{(s)}\rangle) \le \varepsilon.
\]
In practice we estimate the maximum by mini-batching the training states, which is sufficient for our purposes.
\end{definition}
The redundancy graph inside $S_k$ contains an edge $(i,j)$ whenever the estimated distance $\widehat{d}_{\mathrm{FS}}(O_i,O_j)$ does not exceed $\varepsilon$. The connected components of this graph then define candidate clusters of redundant gates.
\begin{theorem}[Redundancy completeness inside subgroups]\label{thm:completeness}
Suppose that for every pair of gates $O_i,O_j \in S_k$ that are $\varepsilon$-redundant with respect to $\mathcal{D}$, the estimated distance obeys $\widehat{d}_{\mathrm{FS}}(O_i,O_j) \le \varepsilon$ and hence an edge $(i,j)$ is inserted in the redundancy graph. Then $O_i$ and $O_j$ lie in the same connected component of the graph and will be merged (directly or transitively) by Algorithm~\ref{alg:lieprune}.
\end{theorem}
\begin{proof}
Under the stated assumption the redundancy graph contains an edge between any pair of $\varepsilon$-redundant gates in $S_k$. The connected components of this graph are precisely the equivalence classes induced by the transitive closure of the pairwise relation ``$\varepsilon$-redundant''. Therefore any two $\varepsilon$-redundant gates must belong to the same connected component and are treated as a single cluster by the pruning procedure.
\end{proof}
For each connected component $C_k$ we select a core gate $O_{\mathrm{core}}$ with the largest Lie sensitivity $S(O_i) = \|\nabla_{X_i} L\|$ and merge all other gates in $C_k$ into a single effective generator
\begin{equation*}
    \begin{split}
X_{\mathrm{new}} = X_{\mathrm{core}} + \sum_{O_m \in C_k \setminus \{O_{\mathrm{core}}\}} \alpha_m X_m,\\ \qquad
\alpha_m = \frac{S(O_m)}{\sum_{O_\ell \in C_k} S(O_\ell)}.
    \end{split}
\end{equation*}
The corresponding gate $O_{\mathrm{new}} = e^{X_{\mathrm{new}}}$ replaces the entire component. The following theorem quantifies the approximation error introduced by this merge.
\begin{theorem}[Approximate functional preservation after pruning]\label{thm:functional}
Consider a redundancy component $C_k$ inside a subgroup $S_k$ and let $U_{\mathrm{orig}}$ and $U_{\mathrm{new}}$ denote, respectively, the unitary implemented by the original sequence of gates in $C_k$ and by the merged gate $O_{\mathrm{new}}$. Assume that
\begin{enumerate}
\item all gates $O_m \in C_k$ are $\varepsilon$-redundant with respect to $\mathcal{D}$ in the sense of \cref{def:redundancy};
\item their generators satisfy $\|[X_m,X_{m'}]\| \le \eta$ for all $m,m' \in C_k$;
\item the weights $\{\alpha_m\}$ are non-negative and obey $\sum_{m \in C_k} \alpha_m \le 1$.
\end{enumerate}
Then there exist constants $C_1,C_2 > 0$, independent of $C_k$, such that for every input state $|\psi\rangle \in \mathcal{D}$ the corresponding output states $|\Psi_{\mathrm{orig}}\rangle = U_{\mathrm{orig}}|\psi\rangle$ and $|\Psi_{\mathrm{new}}\rangle = U_{\mathrm{new}}|\psi\rangle$ satisfy
\[
d_{\mathrm{FS}}(|\Psi_{\mathrm{new}}\rangle,|\Psi_{\mathrm{orig}}\rangle)
\le C_1 |C_k|\,\varepsilon + C_2 |C_k|^2 \eta \;\equiv\; \Delta_{\max}.
\]
\end{theorem}
\begin{proof}
Write the original subcircuit as $U_{\mathrm{orig}} = \prod_{m\in C_k} e^{X_m}$. By repeated application of the Baker--Campbell--Hausdorff formula and the small-commutator assumption we can express
\[
U_{\mathrm{orig}} = e^{X_{\mathrm{avg}} + R},
\]
where $X_{\mathrm{avg}}$ is a convex combination of the generators $\{X_m\}_{m\in C_k}$ and the remainder $R$ collects nested commutators of order at least two. Standard estimates for BCH series on bounded sets imply $\|R\| \le c_2 |C_k|^2 \eta$ for some constant $c_2>0$. On the other hand, the $\varepsilon$-redundancy condition and local Lipschitz continuity of the matrix logarithm (in a small neighborhood of $X_{\mathrm{core}}$) yield $\|X_{\mathrm{new}} - X_{\mathrm{avg}}\| \le c_1 |C_k| \varepsilon$ for some $c_1>0$. Combining the two bounds and using Lipschitz continuity of the Fubini--Study distance with respect to the operator norm of the implementing unitary, we obtain
\[
d_{\mathrm{FS}}(|\Psi_{\mathrm{new}}\rangle,|\Psi_{\mathrm{orig}}\rangle)
\le C_1 |C_k|\,\varepsilon + C_2 |C_k|^2 \eta
\]
for suitable constants $C_1,C_2>0$ that depend only on the dimension of the local Hilbert space.
\end{proof}
In all our experiments the components $C_k$ remain small and the hyperparameters $(\varepsilon,\eta)$ are chosen such that $\Delta_{\max}$ lies safely below the task-specific tolerance, which is corroborated empirically by the negligible accuracy drop shown in \cref{tab:lieprune_three_datasets}.
\begin{theorem}[Complexity under bounded local degree]\label{thm:complexity}
Assume that inside each subgroup $S_k$ the redundancy graph constructed above has maximum degree at most $D$, i.e.\ each gate is connected to at most $D$ candidate neighbors. Then the total running time of LiePrune is
\[
T(N) = O(N d^2 + N D),
\]
which is linear in the number of gates $N$ for fixed Lie-algebra dimension $d$ and bounded local degree $D$.
\end{theorem}
\begin{proof}
Computing the Lie-algebra generators $\{X_i\}$ and their features requires $O(N d^2)$ operations. Within each subgroup $S_k$ every gate $O_i$ is compared to at most $D$ neighbors when constructing the redundancy graph, so at most $ND$ accelerated distance evaluations are performed overall. Graph traversal and connected-component extraction cost $O(N + E)$, where $E$ is the number of edges and $E \le ND$ by the degree bound. Summing all contributions gives $T(N) = O(N d^2 + ND)$.
\end{proof}
Empirically we observe that $D$ is small (typically $D\le 5$) across all benchmarks, which explains the nearly linear scaling of pruning time with circuit size observed in our experiments.
	
\begin{algorithm}[tb]
	\caption{LiePrune One-Shot Structured Pruning}
	\label{alg:lieprune}
	\begin{algorithmic}
		\State Partition gates into Lie subgroups $\{S_k\}$
		\For{each subgroup $S_k$}
		\State Compute accelerated $d_{\mathrm{FS}}$ within $S_k$
		\State Build redundancy graph and extract connected components
		\State Select core gate with maximum Lie sensitivity $S(O_i) = \|\nabla_{X_i} L\|$
		\State Merge all other gates in component via weighted Lie algebra addition
		\EndFor
		\State Output pruned circuit
	\end{algorithmic}
\end{algorithm}
	
\section{Experiments}
	
\subsection{Experimental Setup}
We evaluate LiePrune on three representative QML benchmarks under a unified discriminative setting:
(i) MNIST 4-vs-9 classification with 8 qubits and amplitude embedding,
(ii) FashionMNIST Sandal-vs-Boot with 10 qubits,
and (iii) Bars-and-Stripes generation/discrimination with 8 qubits.
All three circuits are implemented in PennyLane with the JAX backend.
The ansatz consists of 12 hardware-efficient layers in all these experiments.
In addition, we consider a variational quantum eigensolver (VQE) task on the LiH molecule (STO-3G basis) at a fixed bond length $R=1.60$\ \AA.
The electronic-structure Hamiltonian is generated with PennyLane-QChem, and we use the same 12-layer hardware-efficient ansatz on \mbox{$12$ qubits} to approximate the ground-state energy.
For the three discriminative benchmarks, we report test accuracy, whereas for LiH VQE, we report the estimated ground-state energy (in Hartree).
\subsection{Main Results}
\begin{table*}[t]
    \centering
    \caption{LiePrune results on three QML classification benchmarks (seed=42).}
    \label{tab:lieprune_three_datasets}
    \begin{tabular}{lccccccc}
        \toprule
        Dataset & Qubits & Acc$_\text{orig}$ & Acc$_\text{noFT}$ & Acc$_\text{+FT}$ &
        Params$_\text{orig}$ & Params$_\text{pruned}$ & Left(\%) / Comp. \\
        \midrule
        MNIST 4 vs 9 &
        8 &
        1.000 & 0.562 & 0.959 &
        288 & 36 &
        12.5 / 8.0$\times$ \\
        
        Fashion Sandal vs Boot &
        10 &
        0.847 & 0.605 & 0.740 &
        360 & 36 &
        10.0 / 10.0$\times$ \\
        
        Bars \& Stripes &
        8 &
        0.743 & 0.448 & 0.751 &
        288 & 36 &
        12.5 / 8.0$\times$ \\
        \bottomrule
    \end{tabular}
\end{table*}

\begin{table*}[t]
    \centering
    \caption{LiePrune on LiH VQE at $R=1.60$~\AA\ (STO-3G, seed=42).}
    \label{tab:lieprune_lih_vqe}
    \begin{tabular}{lcccccc}
        \toprule
        Task & Qubits & Layers &
        $E_\text{orig}$ & $E_\text{noFT}$ & $E_\text{+FT}$ &
        Left(\%) / Comp. \\
        \midrule
       LiH VQE &
       12&
       12 &
       -7.5225 &
       -3.7416 &
       -4.2875 &
       8.33 / 12.0$\times$ \\
        \bottomrule
    \end{tabular}
\end{table*}
	
\begin{figure*}[!t]
    \centering
    \includegraphics[width=0.88\linewidth]{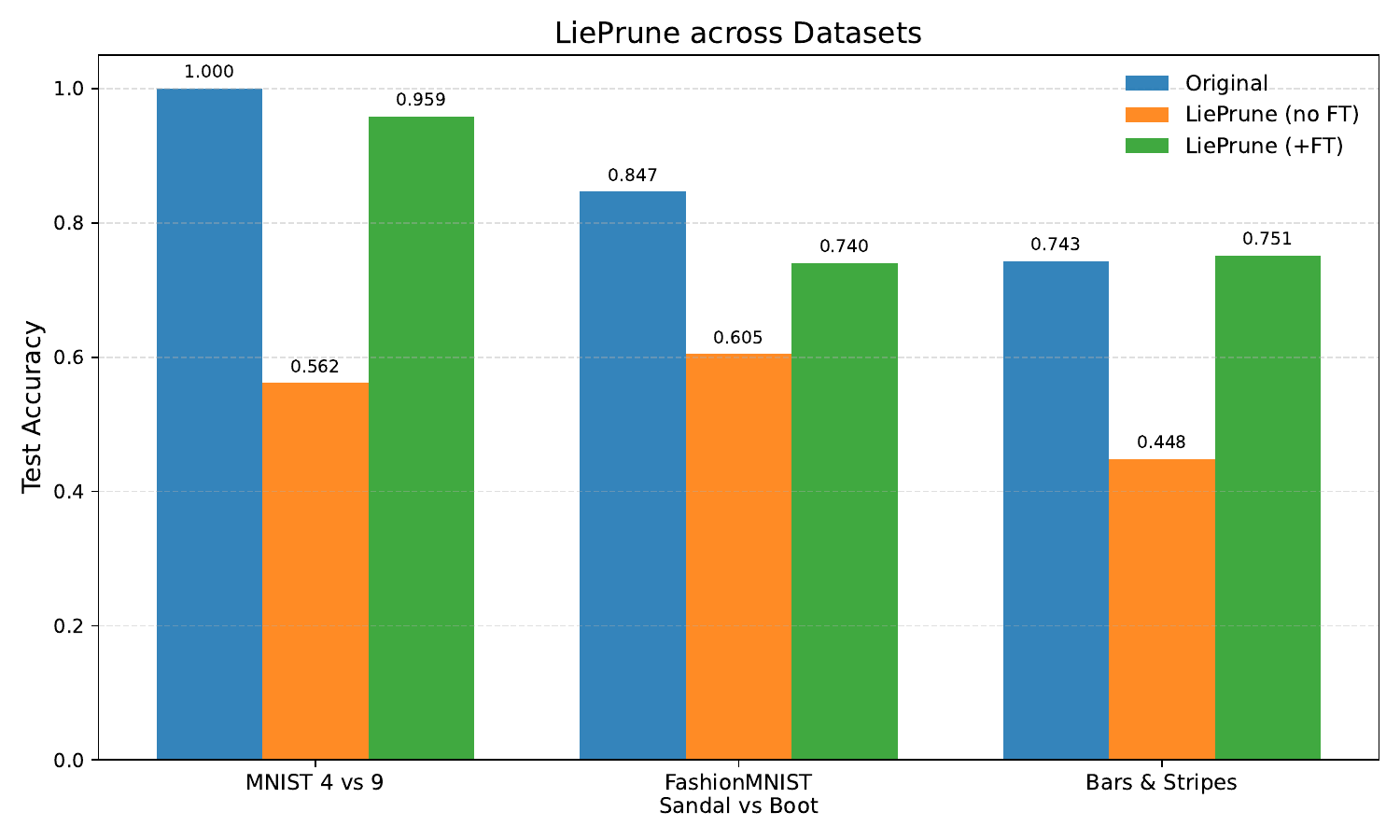}
    \caption{LiePrune across datasets. Each group corresponds to one dataset and contains three bars: Original, LiePrune (no FT), and LiePrune (+FT).}
    \label{fig:lieprune_all_datasets_acc}
\end{figure*}

\begin{figure*}[!t]
    \centering
    \includegraphics[width=0.78\linewidth]{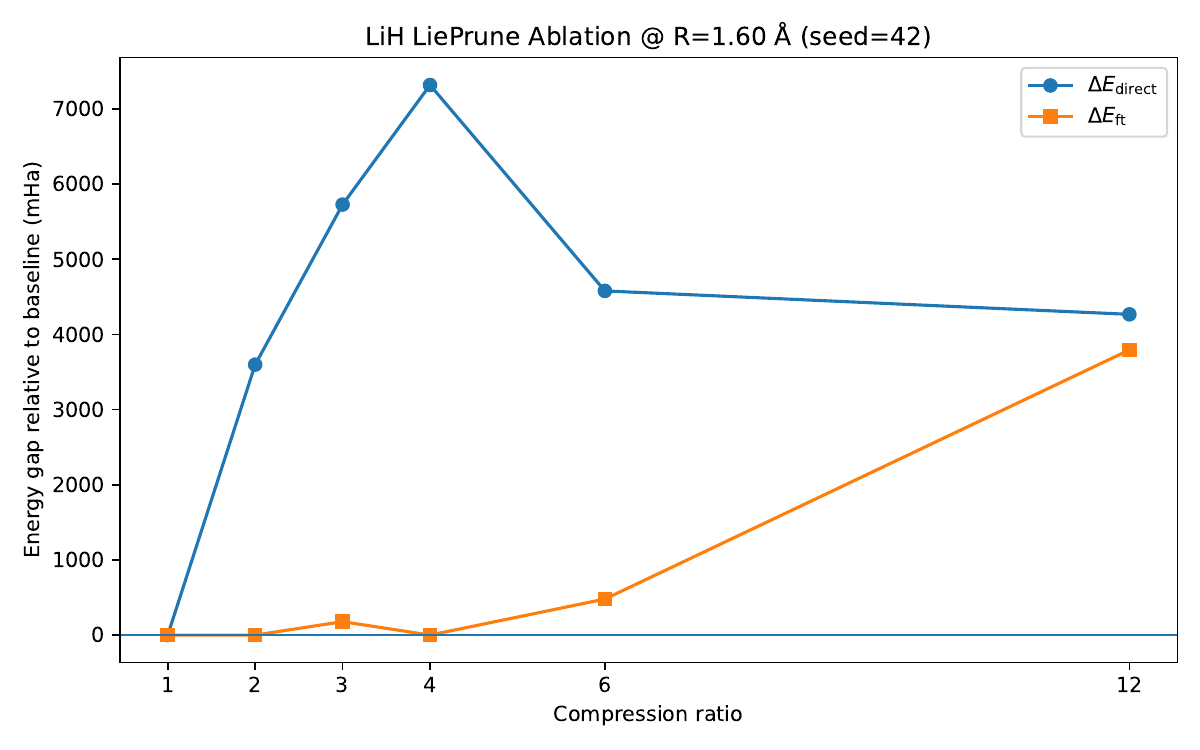}
    \caption{
        LiePrune on LiH VQE at $R=1.60$~\AA\ across multiple compression ratios.
        We report the deviation from the unpruned baseline for direct pruning
        ($\Delta E_\text{direct}$) and post-finetuning ($\Delta E_\text{ft}$).
        Mild compression is recoverable, whereas aggressive compression produces
        multi-Hartree errors that fine-tuning cannot fully mitigate.
    }
    \label{fig:lih_lieprune_ablation}
\end{figure*}

\begin{table*}[!t]
\centering
\caption{
Energy deviation (mHa) under different LiePrune compression ratios for LiH at
$R=1.60$~\AA\ (12 qubits, 16 layers, seed=42). Direct pruning error is
$\Delta E_{\text{direct}}$ and post-finetuning error is
$\Delta E_{\text{ft}}$.
}
\label{tab:lih_vqe_multi_comp}
\begin{tabular}{cccccc}
\toprule
Compression & Groups & Params & $\Delta E_{\text{direct}}$ (mHa) &
$\Delta E_{\text{ft}}$ (mHa) \\
\midrule
$1\times$  & 12 & 576 & $0.02$ & $0.00$ \\
$2\times$  & 6  & 288 & $3596.97$ & $-0.02$ \\
$3\times$  & 4  & 192 & $5726.06$ & $178.44$ \\
$4\times$  & 3  & 144 & $7315.87$ & $0.00$ \\
$6\times$  & 2  & 96  & $4577.52$ & $478.97$ \\
$12\times$ & 1  & 48  & $4265.90$ & $3792.95$ \\
\bottomrule
\end{tabular}
\end{table*}

Figure~\ref{fig:lieprune_all_datasets_acc} and Table~\ref{tab:lieprune_three_datasets}
summarize the performance and compression statistics of LiePrune on the three discriminative benchmarks.
On MNIST 4-vs-9, LiePrune reduces parameters from 288 to 36 (12.5\%, 8$\times$) and recovers most of the accuracy after short fine-tuning (1.000 $\rightarrow$0.562 $\rightarrow$0.959).
On FashionMNIST Sandal-vs-Boot, we observe a similar compression pattern (360 to 36, 10.0\%, 10$\times$) with moderate post-pruning recovery (0.847 $\rightarrow$0.605 $\rightarrow$0.740).
On Bars-and-Stripes, LiePrune achieves 8$\times$ compression (288 to 36) with accuracy slightly improving after fine-tuning relative to the original model (0.743 $ rightarrow$0.448 $\rightarrow$0.751).
These results indicate that LiePrune can substantially reduce parameters with task-dependent accuracy recovery, suggesting the necessity of dataset-aware fine-tuning configurations.

For the LiH VQE task, we first report the baseline 12-qubit, 12-layer result in Table~\ref{tab:lieprune_lih_vqe}, where LiePrune at a $12\times$ compression reduces parameters from 432 to 36 (8.33\% remaining). Unlike the classification benchmarks, pruning here leads to a substantial energy deviation: $E_\text{orig}=-7.5225$~Ha deteriorates to $E_\text{noFT}=-3.7416$~Ha, and fine-tuning recovers only partially to $E_\text{+FT}=-4.2875$~Ha (gaps:$+3.78$~Ha and $+3.23$~Ha, respectively).

To better understand the behavior of LiePrune on chemically structured
Hamiltonians, we further evaluate multiple compression ratios using the
16-layer, 12-qubit ansatz (Figure~\ref{fig:lih_lieprune_ablation} and
Table~\ref{tab:lih_vqe_multi_comp}). We observe that LiePrune behaves differently depending on the compression level. At low compression ($1\times$–$2\times$), pruning induces only minor energy deviations, which are fully recoverable after fine-tuning. In contrast, aggressive compression ($3\times$–$12\times$) leads to large, often multi-Hartree errors, and fine-tuning is generally insufficient to restore the original ground-state energy. These results indicate that chemically structured VQE ansatzes are highly sensitive to high compression, and preserving expressivity under strong pruning requires chemistry-aware strategies or symmetry-preserving parameterizations.

\section{Discussion and Conclusion}
We propose LiePrune, a one-shot structured pruning framework that leverages the Lie-algebraic structure of parameterized quantum circuits.
In this revised experimental scope, we validate LiePrune on three discriminative QML benchmarks and one quantum-chemistry VQE task, all implemented in PennyLane with a JAX backend.
The classification results in Figure~\ref{fig:lieprune_all_datasets_acc} and Table~\ref{tab:lieprune_three_datasets} show that LiePrune consistently achieves 8$\times$--10$\times$ parameter compression by sharing a single effective parameter set per layer across qubits, while maintaining competitive post-fine-tuning accuracy, albeit with task-dependent degradation.
In contrast, the LiH VQE experiment in Table~\ref{tab:lieprune_lih_vqe} reveals that the same one-shot merging strategy can induce a large energy increase (about 3.2--3.8~Ha), even after fine-tuning, despite a more aggressive 12$\times$ compression.
This discrepancy suggests that chemically structured Hamiltonians are substantially more sensitive to subgroup-wise gate merging than the considered classification benchmarks, and that VQE-specific constraints or regularizers are needed for LiePrune to be practically useful in quantum chemistry.

A current limitation is that all reported numbers correspond to a single random seed; multi-seed statistics and a more systematic fine-tuning schedule will further strengthen the empirical claims.
Future work will therefore include multi-seed robustness evaluation, task-aware post-pruning optimization, and extensions to larger-scale QML tasks and more realistic quantum hardware noise models.

\bibliographystyle{IEEEtran}   
\bibliography{ref} 

\end{document}